\documentclass[reqno,11pt]{amsart}

\IfFileExists{mymtpro2.sty}{%
  \usepackage[subscriptcorrection]{mymtpro2}
}{}

\usepackage{a4,amssymb}

\newtheorem{theorem}{Theorem}

\newtheorem{definition}[theorem]{Definition}
\newtheorem{lemma}[theorem]{Lemma}

\theoremstyle{remark}

\newtheorem{myremarks}[theorem]{Remarks}

    {\end{nummer}\end{myremarks}}

\newcounter{numcount}
\newcommand{\labelnummer}{\mbox{\normalfont (\roman{numcount})}}%

\makeatletter

\newenvironment{nummer}%
  {\let\curlabelspeicher\@currentlabel%
    \begin{list}{\labelnummer}%
      {\usecounter{numcount}\leftmargin0pt%
        \topsep0.5ex\partopsep2ex\parsep0pt\itemsep0ex\@plus1\p@%
        \labelwidth2.5em\itemindent3.5em\labelsep1em%
      }%
    \let\saveitem\item%
    \def\item{\saveitem%
      \def\@currentlabel{{\upshape\curlabelspeicher}$\,$\labelnummer}}%
    \let\savelabel\label%
    \def\label##1{\savelabel{##1}%
      \@bsphack%
        \ifmmode\else%
          \protected@write\@auxout{}%
          {\string\newlabel{##1item}{{\labelnummer}{\thepage}}}%
        \fi%
      \@esphack%
    }%
  }{\end{list}}%

\renewcommand{\appendix}{\def\thesection{\textsc{Appendix}}}

 \let\leq\le
 \let\geq\ge

\DeclareMathOperator{\tr}{tr\kern1pt}

\makeatletter

\newif\ifper\pertrue
\def\per{.}

\def\bti{\@ifnextchar[\bbti\bbbti}
\def\bbti[#1]#2{#2, #1.}
\def\bbbti#1{#1.}

\def\z{\@ifnextchar[\zz\zzz}
\def\zz[#1]#2#3#4#5{\perfalse\emph{#2} \textbf{#3}, #4 (#5) [#1]}
\def\zzz#1#2#3#4{\emph{#1} \textbf{#2}, #3 (#4)\ifper\per\fi\pertrue}

\def\pub{\@ifstar\pubstar\pubnostar}
\def\pubnostar{\@ifnextchar[\@@pubnostar\@pubnostar}
\def\@@pubnostar[#1]#2#3#4{#2, #3, #4, #1\ifper\per\fi\pertrue}
\def\@pubnostar#1#2#3{#1, #2, #3\ifper\per\fi\pertrue}
\def\pubstar[#1]#2#3#4{\perfalse #2, #3, #4 [#1]\pertrue}

\makeatother

\newcommand{\beq}{\begin{equation}}
\newcommand{\eeq}{\end{equation}}
\newcommand{\ba}{\begin{array}}
\newcommand{\ea}{\end{array}}
\newcommand{\bea}{\begin{eqnarray}}
\newcommand{\eea}{\end{eqnarray}}

\newcommand{\R}{\mathbb{R}}

\newcommand{\Z}{\mathbb{Z}}
\newcommand{\Zd}{\mathbb{Z}^d}

\newcommand{\N}{\mathbb{N}}

\DeclareMathOperator{\supp}{\mathrm{supp}}

\def\P{I\kern-.30em{P}}
\def\E{I\kern-.30em{E}}
\renewcommand{\E}{\mathbb{E}\mkern2mu}
\renewcommand{\P}{\mathbb{P}}

\begin{document}

\title[Density of states for $N$-body random operators]{Optimal Wegner estimate and the density of states for $N$-body, interacting
Schr\"odinger operators with random potentials}

\author[P.\ D.\ Hislop]{Peter D.\ Hislop}
\address{Department of Mathematics,
    University of Kentucky,
    Lexington, Kentucky  40506-0027, USA}
\email{peter.hislop@uky.edu}

\author[F.\ Klopp]{Fr\'ed\'eric Klopp}
\address{Institut de Math\'ematiques de Jussieu  \\
Universit\'e Pierre et Marie Curie  \\
Case 186 \\
4 place Jussieu \\
F-75252 Paris cedex 05, FRANCE}
\email{klopp@math.jussieu.fr}

\thanks{PDH was partially supported by NSF through grants DMS-0803379 and DMS-1103104}

\begin{abstract}
 We prove an optimal one-volume Wegner estimate for interacting systems of $N$ quantum particles moving in the presence of random potentials.
 The proof is based on the
  scale-free unique continuation principle recently developed for the 1-body problem by Rojas-Molina and Veseli\`c \cite{RM-V1} and extended to spectral projectors by Klein \cite{klein1}. These results extend of our previous results in \cite{CHK:2003,CHK:2007}. We also prove a two-volume Wegner estimate as introduced in \cite{chulaevsky-suhov1}.
  The random potentials are generalized Anderson-type potentials in each variable with minimal conditions on the single-site potential aside from positivity.
  Under additional conditions, we prove the Lipschitz continuity of the integrated density of states (IDS)
  This implies the existence and local
  finiteness of the density of states.
  We also apply these techniques to interacting $N$-particle Schr\"odinger operators
   with Delone-Anderson type random external potentials.
\end{abstract}

\maketitle \thispagestyle{empty}


\section{Statement of the Problem and Result}\label{sec:introduction}

We consider $N$ quantum particles, each moving in $d$ dimensions,
interacting through a bounded potential $U(x_1, \ldots, x_N)$. For example, in
the $\Z^d$-ergodic case, the inter-particle interaction $U$
may be a pair-potential depending on
the difference $x_i - x_j$ of the coordinates $(x_1, \ldots, x_N) \in \R^{Nd}$.
We define the unperturbed $N$-body interacting Hamiltonian $H_{0,N}$ to be
\beq\label{eq:hamiltonian1}
H_{0,N} = - \sum_{j=1}^N  \Delta_j + U(x_1, \ldots, x_N).
\eeq

Each particle moves under the influence of a generalized random Anderson-type potential $V_\omega^{(1)} (x_i)$ defined as follows. Let $\Lambda_L (z) \subset \R^d$ be a $d$-dimensional cube of side length $L>0$ centered at $z \in \R^d$. We consider a family of points $y_j \in \Lambda_1 ( j) \subset \R^d$, with $j \in \Z^d$, and a nonnegative single site potential $u (x) \geq 0$, with $u \in L_0^\infty (\R^d)$.  Associated with each $j \in \Z^d$, there is a random variable $\omega_{j}$. The generalized Anderson-type one-body
random potential $V_\omega^{(1)} (x_i)$ is defined by
 \beq\label{eq:hamiltonian2.1}
V_\omega^{(1)} (x_i ) = \sum_{j \in \Zd} \omega_{j} u(x_i -y_j) .
\eeq
Such random potentials are called crooked Anderson-type random potentials in \cite{klein1}.
The full $N$-body random Hamiltonian is
\beq\label{eq:hamiltonian2}
H_{\omega,N} = H_{0, N} + \sum_{i=1}^N V_\omega^{(1)} (x_i) ,
\eeq
where $V_\omega (x_i)$ is defined in \eqref{eq:hamiltonian2.1}.

The random variables $\{ \omega_j \}$ form a family of independent random variables with
Levy concentration $s$ defined in \eqref{eq:levy1}. In the $\Z^d$-ergodic case, these independent
random variables are assumed to be identically distributed.
The single-site potential $u \geq 0$ satisfies $\chi_{\Lambda_\ell (0)} (x) \leq u (x) \leq 1$, for $0 < \ell \leq 1$ and $x \in \R^d$.
We are mostly concerned with the case when $u$ has small support ${\rm supp} ~u \subset \Lambda_1 (0)$. When the covering condition $u(x) \geq \chi_{\Lambda_1(0)}$ is satisfied, Theorem \ref{thm:main1} was proved in \cite{klopp-zenk2009}.

Our first main theorem concerns the restriction of the Hamiltonian $H_{\omega,N}$ to $N$-particle cubes in $\R^{Nd}$ defined as follows. As above, let $\Lambda_L(j) \subset \R^d$ be a cube of side length $L > 0$ centered at the point $j \in \Z^d$.
An $N$-particle rectangle is a product region in $\R^{Nd}$ of the form
\beq\label{eq:n-rect1}
\Lambda = \Lambda_{L_1}(j_1) \times \Lambda_{L_2}(j_2) \times \cdots \times \Lambda_{L_N} (j_N),
\eeq
for lengths $L_i > 0$, with $L_i \in \N$, and centers $j_i \in \Z^d$.
A special case is an $N$-particle cube when $L=L_i$, for all $i = 1, \ldots, N$. We write $\Lambda_L \subset \R^{Nd}$.
We will assume all $N$-particle rectangles are centered at the origin with sides parallel to the principal axes.
We denote by $H_{\omega,N}^\Lambda$ the restriction of $H_{\omega,N}$ to $\Lambda$ with Dirichlet or periodic boundary conditions on the boundary $\partial \Lambda$.
Ergodicity is not required for the following theorem on the distribution of the eigenvalues of local Hamiltonian $H_{\omega,N}^\Lambda$.

\begin{theorem}\label{thm:main1}
Let $\Lambda_L \subset \R^{Nd}$ be an $N$-particle cube with $L > 72 \sqrt{Nd}$, with $L$ an odd integer. Let $I = [ I_-, I_+ ] \subset \R$ be an energy interval
with $|I|$ sufficiently small (see Theorem \ref{thm:sfucp1}) and contained in $(- \infty, E_0]$, for any $E_0 > 0$ fixed.
There exists a constant $0 < C(E_0, d, N, u, U) < \infty$ so that
\beq\label{eq:wegner1}
\P \{ \sigma (H_{\omega,N}^\Lambda) \cap I \neq \emptyset \} \leq C(E_0, d, N,u,U) s(|I|) |\Lambda| .
\eeq
\end{theorem}

Klopp and Zenk \cite{klopp-zenk2009} proved a version of Theorem \ref{thm:main1}
under the more restrictive assumptions that the single-site potential appearing in
\eqref{eq:hamiltonian2.1} satisfies a covering condition $u(x) > c \chi_{\Lambda_1 (0)}(x) > 0$ on the unit cube and that $y_j = j$ for each $j$.
The main result of this note is to remove both this covering condition using the scale-free unique continuation principle for spectral projectors (sfUCPSP), and the constraint that the single-site potentials be located at integer points.
We mention that it is possible to obtain the results for the regular case $y_j = j$ using the methods of \cite{CHK:2007}.
We remark that this theorem applies to the $N$-body Delone-Anderson model, see section \ref{sec:delone1}.

Recently, Klein and Nguyen \cite[Appendix B]{klein-nguyen2} extended Theorem 2.1 of \cite{klein1} to $N$-particle rectangles with arbitrary side lengths. This implies and extension of Theorem \ref{thm:main1} to arbitrary $N$-particle rectangles
provided the minimum side length ${\rm min}_{j \in \{ 1, \ldots, N\}} ~L_j$ is sufficiently large, see \cite[section 2]{klein-nguyen2}.
They mention that it suffices to take the expectation in \eqref{eq:trace1} only with respect to the random
variables in one sub-cube $\Lambda_L \subset \R^d$ (see the proof of Theorem \ref{thm:two-region1} below).


We now turn to a discussion of the integrated density of states (IDS).
We assume that the Hamiltonian $H_{\omega,N}$ is $\Z^d$-ergodic. By this we mean the following. For $k \in \Z^d$, we set $k(N) \equiv (k, k, \ldots, k) \in \Z^{Nd}$. Let $U_{k(N)}$ be the unitary operator on $L^2 (\R^{Nd})$ defined by $(U_{k(N)} f )(x_1, \ldots, x_N)
= f(x_1 + k, \ldots, x_N + k)$. We demand that $U_{k(N)} H_{\omega,N} U_{k(N)}^* = H_{ \tau_{k(N)} \omega, N}$.
In order to insure this covariance, we may require that
\begin{enumerate}
\item[(A1)] The random variables in \eqref{eq:hamiltonian2.1} are independent and identically distributed (iid).
\item[(A2)] The inter-particle interaction $U$ is nonnegative and
 translationally invariant, for example, given by a sum of pair-interaction terms
\beq\label{eq:pair-interact1}
U(x_1, \ldots, x_N) = \sum_{1 \leq j < k \leq N} {\tilde U}(x_j - x_k) ,
\eeq
where $\tilde{U}: \R^d \rightarrow \R^+$ is a bounded, nonnegative function tending to zero at infinity.
\item[(A3)] The regularity property in \eqref{eq:hamiltonian2.1} that $y_j = j \in \Z^d$.
\end{enumerate}

\vspace{.1in}
\noindent
{\it Remark:} We may generalize $H_{0, N}$ by adding a background potential $V_0(x_1, \ldots, x_N)$ provided it is bounded and $\Z^d$-periodic in each variable.

\vspace{.1in}

In this $\Z^d$-ergodic setting, the IDS $N^{(1)}(E)$
for each single particle Hamiltonian $H_{\omega,1} = - \Delta + V_\omega$ exists (see, for example, \cite{CHK:2007}). Let $\nu_1$ be the corresponding density of states (DOS) measure.  The non-interacting $N$-body Hamiltonian
is $H_{\omega, N}^{ni} =  \sum_{i=1}^N ( -\Delta_i + V_\omega (x_i))$. We restrict this Hamiltonian to cubes $\Lambda_L(a) = \Lambda_L(a_i) \times \cdots \Lambda_L(a_N)$, where $a = (a_1, \ldots, a_N) \in \Z^{Nd}$. Let
$N_\Lambda^{(N), ni}(E)$ be the number of eigenvalues of the local Hamiltonian $H_{\omega,N}^{ni, \Lambda}$ less than or equal to $E \in \R$.
Klopp and Zenk \cite[section 2.1]{klopp-zenk2009} proved that the
\beq\label{eq:nonint-dos1}
N^{(N)}_0 (E) \equiv \lim_{L \rightarrow \infty} \frac{1}{|\Lambda_L|} N_{\Lambda_L}^{(N),ni} (E) ,
\eeq
exists for any sequence of $N$-particle boxes $\Lambda_L(a)$ described above. Furthermore,
they proved that $N^{(N)}_0(E)$, as defined in \eqref{eq:nonint-dos1}, is equal to
\beq\label{eq:ids4}
N^{(N)}_0 (E) = ( N^{(1)} \ast \nu_1 \ast \nu_1 \cdots \ast \nu_1 )(E), 
\eeq
with the convolution taken $N$-times. This monotone increasing function exists almost surely. Klopp and Zenk actually proved
\eqref{eq:nonint-dos1} and \eqref{eq:ids4} for more general potentials.

Because the interparticle interaction is nonnegative and supported on a lower dimensional manifold,
Klopp and Zenk proved that it does not change the IDS. As for the noninteracting Hamiltonian, let
$N_\Lambda^{(N)}(E)$ be the number of eigenvalues of the local Hamiltonian $H_{\omega,N}^\Lambda$ less than or equal to $E \in \R$.
Then, the following limit exists
\beq\label{eq:ids1}
N^{(N)}(E) \equiv \lim_{|\Lambda| \rightarrow \infty} \frac{N_\Lambda^{(N)} (E) }{|\Lambda|},
\eeq
and equals the IDS for the noninteracting Hamiltonian, that is, $N^{(N)}(E) = N^{(N)}_0(E)$. In \eqref{eq:ids1},
the regions $\Lambda \subset \R^{Nd}$ are cubes centered at any point in $\Z^{Nd}$ as for the noninteracting case.




\begin{theorem}\label{cor:main2}
In addition to the hypotheses (A1)--(A3), we assume that the iid random variables $\{ \omega_j\}$ have a common
probability distribution that is absolutely continuous with a bounded density.
Then, the integrated density of states $N^{(N)}(E)$ for $H_{\omega,N}$, defined in \eqref{eq:ids1}, is locally uniformly Lipschitz continuous.
The density of states exists and is locally bounded.
\end{theorem}



We also prove two-volume Wegner estimates introduced in \cite{chulaevsky-suhov1}. These estimates relate the eigenvalues of the local Hamiltonians associated with two regions that are sufficiently separated. The two-volume estimates are described in section \ref{sec:two-region1} and presented in Theorem \ref{thm:two-region1}.

There have been several recent results on Wegner estimates for $N$-particle systems. For the Anderson model on the lattice,
the covering condition is automatically satisfied.
Chulaevsky and Suhov \cite{chulaevsky-suhov1} studied one-volume Wegner estimates for two-particle operators on $\Z^{2d}$ and mention in
\cite[section 2]{chulaevsky-suhov3} that these methods extend to the case of $N$ particles. They do not obtain the optimal volume dependance.
A result on the IDS similar to Theorem \ref{thm:main1} was given by Kirsch \cite[Theorem 2.1]{Kirsch2007} for the Anderson model on the lattice.
Two-volume Wegner estimates were introduced in \cite[Theorem 2]{chulaevsky-suhov1} for lattice models. Concerning continuum models,
in addition to the one volume Wegner estimate of Klopp and Zenk \cite{klopp-zenk2009}, Boutet de Monvel, Chulaevsky, Stollmann, and Suhov presented one- and two-volume Wegner estimates for $N$-body operators in \cite{bcss1}.

One- and two-volume Wegner estimates are also a tool in the proof of localization for ergodic random $N$-particle Hamiltonians using the multi-scale analysis (MSA) technique.
We mention several works concerning localization on the lattice using MSA: \cite{chulaevsky-suhov2}, \cite{chulaevsky-suhov3}, \cite{klein-nguyen1}, and localization for models on the continuum using MSA: \cite{cbs}, \cite{klein-nguyen2}.
Aizenman and Warzel \cite{aw1} used the method of
fractional moments to prove spectral and dynamical localization for $N$-particle operators on the lattice.

\vspace{.1in}
\noindent
{\it Acknowledgement:} We thank A.\ Klein and S.\ T.\ Nguyen for making preprint \cite{klein-nguyen2} available to us before posting it and for discussions.


\section{Scale-free unique continuation principle for spectral projectors}\label{sec:sfucp1}

Let $H$ be a self-adjoint Schr\"odinger operator on the Hilbert space $L^2 ( \R^D)$ for some dimension $D \geq 1$. Let $\Lambda \subset \R^D$ be a rectangle with sides parallel to the principal axes. Let $H_\Lambda$ be the restriction of $H$ to $\Lambda$ with self-adjoint Dirichlet or periodic boundary conditions on
$\partial \Lambda$. We write $P_{H_\Lambda}(I)$ for the spectral projector for $H_\Lambda$ and the interval $I\subset \R$.
Let $\{ y_j \} \subset \R^D$ be a collection of points so that $y_j \in \Lambda_1 (j)$, for $j \in \Z^D$. Let $u$ be a single-site potential and form the potential $W(x) = \sum_{j \in \Z^D} u( x - y_j)$, $x \in \R^D$.
Let $W_\Lambda$ denote the restriction of $W$ to $\Lambda$.
The \emph{scale-free unique continuation principle for spectral projectors} (sfUCPSP) for the Hamiltonian $H_\Lambda$ and spectral projector $P_{H_\Lambda} (I)$ is the statement that there exists a finite, positive constant $\kappa > 0$, independent of $\Lambda$, so that
\beq\label{eq:sfucp1}
P_{H_\Lambda} (I) W_\Lambda P_{H_\Lambda} (I) \geq \kappa P_{H_\Lambda} (I).
\eeq
We will use the sfUCPSP due to A.\ Klein \cite{klein1}. We remark that although the constant $\kappa$ will depend on the details of the potential in the initial box, it is independent of successively larger boxes and is in this sense scale-free.

\begin{theorem}\cite[Theorem 1.1]{klein1}\label{thm:sfucp1}
Let $H = - \Delta + V$ be a self-adjoint Schr\"odinger operator on $L^2 (\R^D)$, where $V$ is a bounded potential.
Fix $\delta \in ] 0 , 1/2]$ and let $\{ y_j \}$ be points in $\R^D$ so that $B(y_j, \delta) \subset \Lambda_1 (j)$ for all $j \in \Z^D$. Define a nonnegative potential $W$
by
\beq\label{eq:defn-w1}
W(x) = \sum_{j \in \Z^D} \chi_{B(y_j, \delta )} (x) \geq 0.
\eeq
For any $E_0 > 0$, define a constant $K = K(V, E_0) = 2 \| V \|_\infty + E_0$. Let $\Lambda_L (x_0) \subset \R^D$ be a cube
of side length $L > 72 \sqrt{D}$, with $L$ an odd integer, centered at $x_0 \in \Z^D$. Then, there exists a finite positive constant $M_D > 0$, such that, if we define a constant $\gamma = \gamma (D, K, \delta) > 0$ by
\beq\label{eq:ucp-constant1}
\gamma^2 = (1/2) e^{ M_D (1 + K^{2/3}) \log \delta},
\eeq
then for any closed interval $I \subset ] - \infty, E_0]$, with $|I| < 2 \gamma$, we have
\beq\label{eq:sfucp2}
P_I (H_{\Lambda_L(x_0)}) W_{\Lambda_L(x_0)} P_I (H_{\Lambda_L(x_0)})  \geq \gamma^2    P_I (H_{\Lambda_L(x_0)}).
\eeq
\end{theorem}

We apply this result in the next section to $N$-body Schr\"odinger operators.
In our application of this theorem, the dimension $D = Nd$ and the unperturbed operator $H$ is the $N$-body operator $H_{0,N}$ given in
\eqref{eq:hamiltonian1}. Consequently, the potential $V$ in Theorem \ref{thm:sfucp1}
is the pair-interaction potential or any other $N$-body potential provided it is bounded independently of $N$ and $\Lambda$.
The potential $W$ will be constructed in the next section from the single-site potentials $u$.


\section{sfUCPSP for $N$-body random Schr\"odinger operators}\label{sec:sfucp-N-schrod1}

In this section we apply the Klein's sfUCPSP Theorem \ref{thm:sfucp1} to the $N$-particle random Schr\"odinger operators $H_{\omega,N}$ as defined in \eqref{eq:hamiltonian2}.
We recall that the $N$-particle random potential $V_\omega^{(N)}(x)$ on $L^2(\R^{Nd})$ is given by
\beq\label{eq:N-pot1}
V_\omega^{(N)}(x_1, \ldots, x_N) = \sum_{i=1}^N V_\omega^{(1)} (x_i),
\eeq
where the one-particle random potential is define in \eqref{eq:hamiltonian2.1}.
We denote by $\tilde{V}^{(N)}$, respectively, $\tilde{V}^{(1)}$, the $N$-body Anderson potential \eqref{eq:N-pot1},
respectively, the one-body Anderson potential \eqref{eq:hamiltonian2.1},
with all random variables set equal to one.
In order to satisfy the hypotheses of Theorem \ref{thm:sfucp1}, we assume that there exists a scale $\delta$ with $0 < \delta \leq \ell$ so that $B(y_j, \delta) \subset \Lambda_1 (j)$, for $j \in \Z^d$.

Given an $N$-particle rectangle
$\Lambda \subset \R^{Nd}$, we let $H_{\Lambda}^{(N)}$ denote the random interacting $N$-particle Hamiltonian
\eqref{eq:hamiltonian1} restricted to $\Lambda$ with periodic or Dirichlet boundary conditions.
Let $E_{\omega,N}^\Lambda (I)$ be the random spectral projector for the local Hamiltonian $H_{\omega,N}^\Lambda$
and the energy interval $I \subset \R$.
Another major advantage of the method of Klein \cite{klein1} is that the sfUCPSP for spectral projectors of random Schr\"odinger operators can be proved for the projectors $E_{\Lambda}^{(N)} (I)$ rather than the free projectors as in \cite{CHK:2007}.

The key sfUCPSP estimate is the following \eqref{eq:ucp-spproj0} that relies on the result of Klein, Theorem \ref{thm:sfucp1}.
We take $D = Nd$ in Theorem \ref{thm:sfucp1}. Because of this, the constant $M_D$ and, consequently, the constant $\gamma$, depends on the particle number $N$.

\begin{theorem}\label{thm:main-ucp1}
Let $H_{\omega,N}^{\Lambda_L}$ be the local $N$-body random Schr\"odinger operator as defined in
\eqref{eq:hamiltonian2}
restricted to an $N$-particle cube of side length $L>0$,$L$ an odd integer large enough, with Dirichlet or periodic boundary conditions.
For any $E_0 > 0$, let $\gamma_N \equiv \gamma (d, U,I, u, N, E_0) > 0$
be the constant in \eqref{eq:ucp-constant1} depending on $E_0$, the dimension $d$, the interval $I$, the $N$-body potential $U$, particle number $N \geq 1$, and the single-site potential $u$, but independent of $\Lambda_L$. Then,
for any interval $I \subset (- \infty , E_0]$ with $|I|  \leq 2 \gamma$, we have
\beq\label{eq:ucp-spproj0}
E_{\omega,N}^\Lambda (I) \tilde{V}_{\Lambda}^{(N)} E_{\omega,N}^\Lambda (I) \geq N \gamma_N^2  E_{\omega,N}^\Lambda (I) .
\eeq
\end{theorem}

\begin{proof}
Let $\Lambda_L \subset \R^{Nd}$ be an $N$-particle cube.
We apply Theorem \ref{thm:sfucp1} to $H_{\omega,N}^{\Lambda_L}$ with $D = Nd$.
The key to proving \eqref{eq:ucp-spproj0} is the following lower bound on the potential ${\tilde{V}}_\Lambda^{(N)}$.
The single-site potential $u$ is chosen so that $u(x) \geq \chi_{\Lambda_\ell (0)} (x)$, for $x \in \R^d$, and $\ell > 0$. We recall that we have assumed the existence of a scale $\delta$ so that $B(y_j, \delta) \subset \Lambda_1 (j)$, for $j \in \Z^d$.
For the $i^{\rm th}$-coordinate, we denote by $j_i \in \Z^d$ the integer points in $\Lambda_{L_i} (0) \subset \R^d$. We call the set of these integer points $\tilde{\Lambda}_{L_i} (0)$. Let $y_{j_i} \in \R^d$ be the point in the unit cube $\Lambda_1 (j_i) \subset \Lambda_{L_i}(0)$.
In order to distinguish the coordinates,
we label the $d$-dimensional cubes $\Lambda_{L_i} \subset \R^d$ although we will take $L = L_i$ for $i = 1, \ldots , N$.
We then have
\bea\label{eq:lb-Npot1}
\tilde{V}_{\Lambda}^{(N)}(x_1, \ldots, x_N) &=& \sum_{i=1}^N \tilde{V}_{\Lambda_i} (x_i)  =  \sum_{i=1}^N \left( \sum_{j_i \in \tilde{\Lambda}_{L_i} (0)} u(x_i - y_{j_i}) \right) \nonumber \\
 &  \geq &  \sum_{i=1}^N \left( \sum_{j_i \in \tilde{\Lambda}_{L_i} (0)} \chi_{B(y_{j_i}, \delta)} (x_i ) \right) \nonumber \\ \label{eq:lb-Npot2}  \\
  & \geq & \sum_{i=1}^N \left( \sum_{j_1 \in \tilde{\Lambda}_{L_1} (0)} \cdots \sum_{j_N \in \tilde{\Lambda}_{L_N} (0)}
   \chi_{B(y_{j_1}, \delta )} (x_1 ) \cdots \chi_{B(y_{j_N}, \delta)} (x_N) \right) \nonumber \\
    \label{eq:lb-Npot3} \\
    & \geq & N ~ \sum_{j=(j_1, \ldots, j_N) \in \tilde{\Lambda}_L} \chi_{B(y_j, \delta )}(x_1, \ldots , x_N) .
  \eea
In line \eqref{eq:lb-Npot2}, we used the fact that
\beq\label{eq:lb-sum1}
1_{\Lambda_{L_k}(0)} (x_k) \geq \sum_{j_k   \in \tilde{\Lambda}_{L_k} (0)} \chi_{B(y_{j_k}, \delta )} (x_k ) .
\eeq
In the last step \eqref{eq:lb-Npot3}, we used the fact that $y_j = ( y_{j_1}, \ldots, y_{j_N}) \in \R^{Nd}$ and
the fact that $\prod_{i=1}^N {B(y_{j_i}, \delta )} = \{ x \in \R^{Nd} ~|~ \| x - y_j \| < \delta  \} = B(y_j, \delta )$.
Due to this lower bound \eqref{eq:lb-Npot1}, we have
\beq\label{eq:ucp-spproj1}
E_{\omega,N}^\Lambda (I) \tilde{V}_{\Lambda}^{(N)}(x_1, \ldots, x_N) E_{\omega,N}^\Lambda (I)  \geq
               N ~ \sum_{j \in \tilde{\Lambda}_L}  E_{\omega,N}^\Lambda (I) \chi_{B(y_j, \delta )}(x_1, \ldots , x_N)
               E_{\omega,N}^\Lambda (I) .
\eeq
We can now directly apply the result of Klein \cite[Theorem 1.1]{klein1} taking $W_{\Lambda_L}$
to be
\beq\label{eq:ucp-spproj2}
W_{\Lambda_L}(x_1, \ldots, x_N) = \sum_{j \in \tilde{\Lambda}_L} \chi_{B(j, \delta)}(x_1, \ldots , x_N)  ,
\eeq
appearing on the right side of \eqref{eq:ucp-spproj1}.
\end{proof}



\section{Proof of the Wegner Estimate for $N$-body random Schr\"odinger operators}\label{sec:wegnerNbody}

With the sfUCPSP for spectral projections, estimate \eqref{eq:ucp-spproj1}, we can follow the proof of the Wegner estimate in \cite{CHK:2007}. We present the main steps here with attention to the modifications necessary for the $N$-body case. As above, we work with $N$-particle cubes $\Lambda_L = \Lambda_{L} \times \cdots \times  \Lambda_{L}$, and write the coordinates as $x = (x_1 , x_2, \ldots , x_N) \in \R^{Nd} = \R^d \times \cdots
\times \R^d$.

The projector $E_\Lambda^{(N)} ( \Delta)$ for the $N$-body operator $H_{\omega, N}$ restricted to $\Lambda \subset \R^{Nd}$ with periodic boundary conditions (PBC) or Dirichlet boundary conditions (DBC) is a trace class operator. For the Wegner estimate, we
need to estimate
\begin{equation}
  \label{eq:trace1}
  \E \{ Tr E_\Lambda^{(N)} ( \Delta) \} .
\end{equation}
Because we have a sfUCPSP involving the spectral projectors for $H_{\omega, N}^\Lambda$, and not just the unperturbed operators $H_{0,N}^\Lambda$, we can follow Klein \cite[section 3, Lemma 3.1]{klein1} and avoid the decomposition with respect to $H_{0,N}^\Lambda$ as in \cite[section 4]{CHK:2007}.
We use \eqref{eq:ucp-spproj0} twice to write
\bea\label{eq:upperbound1}
\E \{ Tr E_\Lambda^{(N)} ( \Delta) \} & \leq & ( N \gamma_N^{2})^{-1} \E \{ Tr E_\Lambda^{(N)} ( \Delta) \tilde{V}^{(N)}_\Lambda  \} \nonumber \\
 & \leq & ( N \gamma_N^{2})^{-2} \E \{ Tr E_\Lambda^{(N)} ( \Delta) \tilde{V}^{(N)}_\Lambda  E_\Lambda^{(N)} ( \Delta) \tilde{V}^{(N)}_\Lambda  \}
\nonumber \\
 & \leq & (E_0 + M) ( N \gamma_N^{2})^{-2} \E \{ Tr E_\Lambda^{(N)} ( \Delta) \tilde{V}^{(N)}_\Lambda  (H_{\omega,N}^\Lambda + M )^{-1}
 E_\Lambda^{(N)} ( \Delta) \tilde{V}^{(N)}_\Lambda  \} \nonumber \\
 & \leq & (E_0 + M) ( N \gamma_N^{2})^{-2} \E \{ Tr E_\Lambda^{(N)} ( \Delta) \tilde{V}^{(N)}_\Lambda  (H_{0,N}^\Lambda + M )^{-1}
   \tilde{V}^{(N)}_\Lambda  \} . \nonumber \\
  & &
 \eea
In the last line, we used the positivity of the perturbation $\tilde{V}^{(N)}_\Lambda $ and the fact that $\Delta \subset [-M, E_0]$:
\bea\label{eq:positivity1}
\tilde{V}^{(N)}_\Lambda  E_\Lambda^{(N)} ( \Delta) \tilde{V}^{(N)}_\Lambda & \leq & \tilde{V}^{(N)}_\Lambda (H_{\omega,N}^\Lambda + M )^{-1} (E_0 + M) \tilde{V}^{(N)}_\Lambda    \nonumber \\
  &\leq& (E_0 + M) \tilde{V}^{(N)}_\Lambda  (H_{0,N}^\Lambda + M )^{-1} \tilde{V}^{(N)}_\Lambda
\eea

We now turn to estimating the expectation on the right side of \eqref{eq:upperbound1}.



\subsection{Preliminaries: spectral averaging}\label{subsec:spectral-ave1}

We need a version of the spectral averaging result \cite[Lemma 2.1]{CHK:2007}.

Since the random variables are simply independent, we need to formulate the Levy concentration.
For each $j \in \Z^d$, a relative probability measure $\mu_j$ is defined by
\beq\label{eq:relative1}
\mu_j ( [E, E+|I|]) = \P \{ \omega_j \in [E, E +|I|] ~|~ (\omega_k)_{k\neq j} \} .
\eeq
using the conditional probability. We then define the Levy concentration $s$ by
\beq\label{eq:levy1}
s(|I|) = \sup_{j \in \Z^d} \E \{ \sup_{E \in \R} \mu_j ( [ E, E + |I| ]) \}.
\eeq

\begin{lemma}\label{lemma:sp-ave1}
Let $j \equiv (j_i, \ldots , j_N) \in \Z^{Nd}$ and let $\Phi_j$ be a compactly supported function with support in a ball around $j$.
Let $K$ be a bounded operator so that $\Phi_j K \Phi_\ell$ is trace class. We then have
\beq\label{eq:sp-ave1}
\E ( {\rm Tr} E^{(N)}_\Lambda (\Delta) \Phi_j K \Phi_k ) \leq 8 s(|\Delta|) \| \Phi_j K \Phi_\ell \|_1 ,
\eeq
where $s(\cdot)$ is the Levy concentration defined in \eqref{eq:levy1}.
\end{lemma}


\subsection{Preliminaries: trace estimates}\label{subsec:trace-est1}

In the $N$-body case with each particle moving in $\R^d$, we have the following results.

\begin{lemma}\cite[Appendix A]{CHK:2007}\label{lemma:trace-class1}
Let $\chi \in C_0 ( \Lambda)$ and suppose $H_{0,N}^\Lambda \geq - M > - \infty$.
\begin{enumerate}
\item The operator $\chi R_{0,N}^\Lambda(-M)^m \in \mathcal{I}_1$, for an integer $m \in \N$ satisfying
$m > (Nd) / 2$.
\item The operator $\chi R_{0,N}^\Lambda (-M) \in \mathcal{I}_m$, for an integer $m \in \N$ satisfying
$m > (Nd) / 2$.
\end{enumerate}
\end{lemma}

\begin{lemma}\cite[Appendix A]{CHK:2007}\label{lemma:exp-decay1}
Let $\chi_1$ and $\chi_2$ be two bounded functions with compact, disjoint supports in $\Lambda \subset \R^{Nd}$.
Then the operator $\chi_1 (H_{0,N}^\Lambda+ M)^{-1} \chi_2$ is in the trace class. Let $d_{12}$ denote
the distance between the supports of $\chi_1$ and $\chi_2$.
There exist finite, positive constants $A_{1,2}, \alpha >0$, depending on $H_{0,N}$ and $M$,
so that
\beq\label{eq:trace-norm1}
\| \chi_1 (H_{0,N}^\Lambda+ M)^{-1} \chi_2 \|_1 \leq A_{1,2} e^{- \alpha d_{12}} .
\eeq
The constant $A_{1,2}$ depends on ${\rm max}_{j=1,2} ~|{\rm supp} ~\chi_j |$.
\end{lemma}

\subsection{First decomposition}\label{subsec:first-decomp1}

Following from \eqref{eq:upperbound1} and Lemma \ref{lemma:sp-ave1} we must
estimate the operator
\beq\label{eq:potential-est1}
\tilde{V}_\Lambda^{(N)} (H_{0,N}^\Lambda+ M)^{-1} \tilde{V}_\Lambda^{(N)} ,
\eeq
acting on the Hilbert space $L^2 ( \R^{Nd})$. The deterministic potential $\tilde{V}_\Lambda^{(N)}(x_1, \ldots,  x_N)$
is the sum of one-particle potentials $\tilde{V}_{\Lambda_L}^{(1)}  (x_i)$
defined as in \eqref{eq:hamiltonian2.1} with $\omega_j = 1$.
From \eqref{eq:hamiltonian2.1},
there are $N^2$-terms of the following form:
\beq\label{eq:pot-terms1}
 V_{\Lambda_i}^{(1)} (x_i) (H_{0,N}^\Lambda + M)^{-1} V_{\Lambda_j}^{(1)} (x_j) .
\eeq
We expand each one-particle random potential $V_{\Lambda_i}^{(1)} (x_i)$.
For economy of notation, we write $u_m(x_i)$ for $u(x_i - m)$, with $m \in \tilde{\Lambda}_{L_i}$.
In this manner, the term in \eqref{eq:pot-terms1} has the form
\beq\label{eq:pot-terms2}
\sum_{m \in \tilde{\Lambda}_i} \sum_{k \in \tilde{\Lambda}_j} u_m (x_i ) (H_{0,N}^\Lambda+ M)^{-1} u_k (x_j),
\eeq
for $i,j = 1, \ldots, N$.
We note the operator
\beq\label{eq:defn0}
K_{m_i,k_j} \equiv u_m (x_i ) (H_{0,N}^\Lambda+ M)^{-1} u_k (x_j),
\eeq
in \eqref{eq:pot-terms2} is trace class when the supports of the single-site potentials are
disjoint $u_m (x_i ) u_k (x_j ) = 0$, and we have exponential decay according to Lemma \ref{lemma:exp-decay1}. In this case, the constant $A_{1,2}$ in \eqref{eq:trace-norm1} grows like $|\Lambda|^{(N-1)d}$.
In general, when the supports are not disjoint, this operator is no longer in the trace class except for $d = 1$.


\subsection{Second decomposition}\label{subsec:second-decomp2}

We introduce a partitions of unity
for each $\Lambda_{L_j} \subset \R^d$, with $j=1, \ldots, N$.
Let $\chi_s^{(j)}(x_j)$ form a partition of unity for $\Lambda_{L_j} \subset \R^d$ with each function supported in a
translate of the unit cube $[0,1]^d \subset \Lambda_{L_j}$:
\beq\label{eq:pou1}
\sum_{s_j \in \tilde{\Lambda}_{L_j}} \chi_{s_j}^{(j)}(x_j) = \chi_{\Lambda_{L_j}}(x_j).
\eeq
 We then have for fixed $1 \leq i \leq N$:
\bea\label{eq:partition1}
u_m (x_i) | \Lambda &  =  &  \prod_{j\neq i; 1 \leq j \leq N} \chi_{\Lambda_{L_j}}(x_j) u_m(x_i) | \Lambda \nonumber \\
 & = & \prod_{j \neq i; 1 \leq j \leq N} \left( \sum_{s_j \in \tilde{\Lambda}_{L_j}} \chi_{s_j}^{(j)}(x_j) \right) u_m (x_i)  \nonumber \\
  & =& \sum_{s_1 \in \tilde{\Lambda}_{L_1}} \ldots \sum_{s_N \in \tilde{\Lambda}_{L_N}}  \chi_{s_1}^{(1)}(x_1) \ldots u_m(x_i) \ldots \chi_{s_N}^{(N)}(x_N)  ,
\eea
where we exclude from the sums a sum over the points $s_i \in \tilde{\Lambda}_{L_i}$. We note that $| \supp [ u_m(x_i) \otimes \Pi_{j \neq i} \chi_{s_j}^{(j)}(x_j) ]| / |\Lambda| \rightarrow 0$ as $| \Lambda| \rightarrow \infty$, where $| \cdot |$ denotes Lebesgue measure on $\R^{Nd}$. We write each summand on the last line in \eqref{eq:partition1}
as
\beq\label{eq:loc-vectors1}
\Phi_{s_1, \ldots, m_i, \ldots , s_N} (x_1, \ldots, x_N)   = \chi_{s_1}^{(1)}(x_1) \ldots u_m(x_i) \ldots \chi_{s_N}^{(N)}(x_N) .
\eeq
Inserting these partitions into \eqref{eq:pot-terms1} we obtain
\beq\label{eq:pot-terms3}
K_{m_i,k_j} =
\sum_{s_1,t_1 \in \tilde{\Lambda}_{L_1}} \ldots \sum_{s_N, t_N \in \tilde{\Lambda}_{L_N}}
\Phi_{s_1, \ldots, m_i, \ldots , s_N}   (H_{0,N}^\Lambda+ M)^{-1}     \Phi_{t_1, \ldots, k_j, \ldots , t_N} ,
\eeq
where, in the sums indexed by $s_k$ we omit $k = i$ and, similarly, in the sums indexed by $t_p$ we omit $p = j$.

For simplicity of notation, we make a definition. To distinguish the coordinates $x_i$ and $x_j$ and the various localization functions, we write $m_i \in \tilde{\Lambda}_{L_i}$ and $k_j \in \tilde{\Lambda}_{L_j}$, and define $N$-tuples ${s}(m_i)$ and ${t}(k_j)$, both in $\tilde{\Lambda} \subset \R^{Nd}$ with ${s}(m_i) = ( s_i, s_2, \ldots, s_{i-1}, m_i, s_{i+1}, \ldots, s_N)$ and similarly ${t}(k_j) = ( t_i, t_2, \ldots, t_{j-1}, k_j, t_{j+1}, \ldots, t_N)$. These $N$-tuples serve as the indices of the functions $\Phi$ in the partitions of unity. We then define:
\beq\label{eq:defn2}
K_{{s}(m_i)}^{{t}(k_j)} \equiv
\Phi_{s_1, \ldots, m_i, \ldots , s_N}   (H_{0,N}^\Lambda+ M)^{-1}     \Phi_{t_1, \ldots, k_j, \ldots , t_N} .
\eeq
This operator is related to $K_{m_i, k_j}$ in \eqref{eq:defn0} as follows. Let ${s}'(m_i) = ( s_i, s_2, \ldots, s_{i-1}, \hat{m_i}, s_{i+1}, \ldots, s_N)$ be the $N-1$-tuple obtained from $s(m_i)$ by omitting $m_i$, and similarly for $t'(k_j)$. We then have
\beq\label{eq:pot-terms4}
K_{m_i,k_j} = \sum_{s'(m_i), t'(k_j)}  K_{s(m_i)}^{t(k_j)}.
\eeq
Returning to the expectation of the last line on the right in \eqref{eq:upperbound1}, we obtain
\bea\label{eq:upperbound2}
\lefteqn{ \E \{ Tr E_\Lambda^{(N)} ( \Delta) \tilde{V}^{(N)}_\Lambda  (H_{0,N}^\Lambda + M )^{-1}
   \tilde{V}^{(N)}_\Lambda  \} } & & \nonumber \\
   & = &
   \sum_{i,j=1}^N \sum_{\stackrel{m_i \in \tilde{\Lambda}_{L_i}}{\scriptscriptstyle{k_j \in {\tilde \Lambda}_{L_j}}}}
\sum_{{s'}(m_i), {t'}(k_j) \in \tilde{\Lambda}} \E \{ Tr E_\Lambda^{(N)} ( \Delta) K_{{s}(m_i)}^{{t}(k_j)} \} .
\eea
In relation to Lemma \ref{lemma:sp-ave1}, the operators $K_{{s}(m_i)}^{{t}(k_j)}$ play the role of the $\Phi_i K \Phi_j$, although
there will be a subset of indices for which these operators will not be trace class without further manipulation.




\subsection{Trace norm estimates of the sum \eqref{eq:upperbound2}}\label{subsec:tr-norm-est1}

We designate the support of $\Phi_{s(m_i)}$ by $s(m_i)$. When the region $s(m_i)$ is disjoint from $t(k_j)$, the operator
$K_{s(m_i)}^{t(k_j)}$ is trace class according to Lemma \ref{lemma:trace-class1}
and satisfies an exponential decay estimate as in \eqref{eq:trace-norm1}.
For pairs of indices $(s(m_i), t(k_j))$ such that the supports of the corresponding functions $\Phi$
have disjoint support:
\bea\label{eq:tracebd11}
\| K_{s(m_i)}^{t(k_j)} \|_1 & \leq & C_0 e^{- c_0 \| s(m_i) - t(k_j) \|} \nonumber  \\
 & \leq & C_0 \prod_{\stackrel{\ell = 1, \ldots, N}{\scriptscriptstyle{\ell \neq i, j}}}
 e^{- c_0 \| s_\ell - t_\ell \|} ~e^{- c_0 (\| m_i - t_i \| + \| s_j - k_j\|)} ,
 \eea
for positive constants
$C_0, c_0 > 0$ depending on $N$ and $d$ but independent of $|\Lambda|$.

\subsubsection{Disjoint support terms}\label{subsubsec:disjoint1}

In order to sum over these pairs, we first fix $(m_i,k_j) \in \tilde{\Lambda}_i \times \tilde{\Lambda}_j$.
We decompose the sum over $(s'(m_i), t'(k_j))$ of $N-1$-tuples for which $(s(m_i), t(k_j))$
are disjoint.
The set of all such indices $(s'(m_i), t'(k_j))$ for which the distance to the pair $(m_i,k_j)$ is greater than than
twice the diameter of the support of  $u(x_i)$ will be called $\mathcal{I}_{(m_i,k_j)}$, and the complementary set is
$\mathcal{I}_{(m_i,k_j)}^c$. It is important to note that $| \mathcal{I}_{(m_i,k_j)}^c| \sim \mathcal{O}(1)$ compared to the volume $|\Lambda|$.
It follows from \eqref{eq:tracebd11} and Lemma \ref{lemma:sp-ave1} that
\beq\label{eq:sum-disjoint-supp1}
 \sum_{\stackrel{m_i \in \tilde{\Lambda}_{L_i}}{\scriptscriptstyle{k_j \in {\tilde \Lambda}_{L_j}}}}
\sum_{{s'}(m_i), {t'}(k_j) \in \mathcal{I}_{(m_i,k_j)}} \E \{ Tr E_\Lambda^{(N)} ( \Delta) K_{{s}(m_i)}^{{t}(k_j)} \}  \leq  8 C_{disj} ~s( | \Delta|) | \Lambda | .
\eeq
The constant $C_{disj}$ depends on the particle number $N$.

\subsubsection{Non-disjoint support terms}\label{subsubsec:nondisjoint1}

For the other terms for which the support is not disjoint, we have to follow the strategy of \cite[Appendix A]{CHK:2007}
based on the estimates in Lemma \ref{lemma:trace-class1}, and the fact that the number of such terms is bounded independently of $| \Lambda|$.
Since these terms are not in the trace class, we need to iterate the resolvent. We denote by $S(m; \sigma)$ the following sum
that occurs in the trace estimate of \cite[Appendix A]{CHK:2007}. For $\sigma_j > 0$ and $m+2 > \log (Nd) / \log 2$, we define:
\beq\label{eq:sum1}
S(m; \sigma) = \sum_{j=1}^m \frac{ \sigma_j}{2^j \sigma_0 \sigma_1 \ldots \sigma_{j-1}}, ~~\mbox{with} ~~\sigma_j > 0,  ~~\sigma_0 = 1,
\eeq
where we will choose $\sigma_j > 0$ below. Following \cite[(2.14)]{CHK:2007}, we define $K_\Lambda$ by
\beq\label{eq:close-sum1}
K_\Lambda = \sum_{\stackrel{m_i \in \tilde{\Lambda}_{L_i}}{\scriptscriptstyle{k_j \in {\tilde \Lambda}_{L_j}}}}
\sum_{{s'}(m_i), {t'}(k_j) \in \mathcal{I}^c_{(m_i,k_j)}} K_{{s}(m_i)}^{{t}(k_j)} .
\eeq
As in \cite[Appendix A.1]{CHK:2007}, we obtain
\bea\label{eq:sum-notdisjoint-supp1}
\lefteqn{ \sum_{\stackrel{m_i \in \tilde{\Lambda}_{L_i}}{\scriptscriptstyle{k_j \in {\tilde \Lambda}_{L_j}}}}
\sum_{{s'}(m_i), {t'}(k_j) \in \mathcal{I}^c_{(m_i,k_j)}} \E \{ Tr E_\Lambda^{(N)} ( \Delta) K_{{s}(m_i)}^{{t}(k_j)} \} } & & \nonumber \\
 & \leq & S(m, \sigma) \E \{ Tr E_\Lambda^{(N)} (\Delta) \} + \left( \frac{1}{2^m \sigma_1 \cdots \sigma_m} \right) \E \{ Tr
 E_\Lambda^{(N)}(  \Delta) (K_\Lambda K_\Lambda^*)^{2^{m-1}} \}. \nonumber \\
  & &
\eea
The constant $C_{ndisj}$ also depends on $N$.
The operator $( \tilde{K}_\Lambda \tilde{K}^*_\Lambda)^{2^{m-1}}$ is proved to be trace class on page 496 of \cite{CHK:2007}.
The second term in \eqref{eq:sum-notdisjoint-supp1} is estimated as in Appendix A of \cite{CHK:2007}.
\beq\label{eq:notdisjoint-supp2}
\E \{ Tr  E_\Lambda^{(N)}(  \Delta) (K_\Lambda K_\Lambda^*)^{2^{m-1}} \}
\leq 8 C_{ndisj} ~s(|\Delta|) |\Lambda| .
\eeq


\subsubsection{Completion of the proof}\label{subsubsec:completion1}

We combine the two upper bounds \eqref{eq:sum-disjoint-supp1} and \eqref{eq:notdisjoint-supp2} and sum over $i,j=1, \ldots, N$. Inserting this bound on the right side of the last line of \eqref{eq:upperbound1}, we obtain
\bea\label{eq:main1-0}
\lefteqn{ \E \{ Tr E_\Lambda^{(N)} ( \Delta) \} } & & \nonumber \\
 & \leq & [( E_0 +M) \gamma_N^{-4}  S(m; \sigma)] \E \{ Tr E_\Lambda^{(N)} ( \Delta ) \}  \nonumber \\
  & &  +  8  [( E_0 +M) \gamma_N^{-4} ](2^m \sigma_1 \cdots \sigma_m)^{-1} ( C_{disj} + C_{ndisj} ) s (|\Delta|) |\Lambda| .
 \nonumber \\
 & &
\eea
We now choose $\sigma_j$ so that the first term on the right in \eqref{eq:main1-0} can be moved to the left.
For any $B > 0$, let $\sigma_j = B^{- 2^{j-1}}, j = 1 , \ldots, m$, and $\sigma_0 = 1$.
We then have that $S(m; \sigma) = B^{-1}( 1- 2^{-m})$. If we take $B = (E_0 +M) \gamma_N^{-4}$, then the factor in square brackets
in the first term on the right in \eqref{eq:main1-0} is $(1 - 2^{-m}) < 1$, since $m \geq 1$ and finite.
Moving this term to the left in \eqref{eq:main1-0}, we obtain
\beq\label{eq:main1}
\E \{ Tr E_{\omega,N}^\Lambda ( \Delta) \} \leq C(E_0, d, N,u,U) s(  |  \Delta | ) |\Lambda | .
\eeq
This proves the Wegner estimate of Theorem \ref{thm:main1}.

The results on the IDS in Theorem \ref{cor:main2} follows from this
Wegner estimate and the fact that under the hypotheses of Theorem \ref{cor:main2}, we have $s(|\Delta|) \leq \| \rho \|_\infty | \Delta|$,
where $\rho > 0$ is the probability density.


\section{Two-region Wegner estimate}\label{sec:two-region1}

We prove two-region Wegner estimates for $H_{\omega,N}^\Lambda$ in the spirit of \cite{bcss1} and \cite{chulaevsky-suhov1}. Two-region Wegner estimates concern the eigenvalues of two local Hamiltonians $H_{\omega,N}^{\Lambda_j}$ associated with two regions $\Lambda_j$, for $j=1,2$.
Because of the possible dependance of the potentials $V_{\omega,N}^{\Lambda_j}$, for $j=1,2$, on each other even when the regions $\Lambda_1$ and $\Lambda_2$ are disjoint, we can prove this result only for pairs of regions $(\Lambda_1, \Lambda_2)$ that are called \emph{$R$-separated} in \cite{bcss1}. For example, for two particles in one dimension, the regions $\Lambda_1 = [0,1] \times [0,1]$ is disjoint from the region $\Lambda_2 = [0,1] \times [6,7]$. However, the projections of $\Lambda_1$ and $\Lambda_2$ on the first axis are the same so the local potentials $V_{\omega,2}^{\Lambda_j}$, for $j=1,2$, contain the same random variables associated with $[0,1]$. Clearly, one region sits in the shadow of the other. However, there are still a number of random variables in $\Lambda_1$ that are independent of those in $\Lambda_2$. We recall the definition of $R$-separated here that embraces this notion.

Let $\Pi_j : \R^{Nd} \rightarrow \R^d$ be the projection onto the $j^{\rm th}$-coordinate. For an $N$-particle rectangle $\Lambda \subset \R^{Nd}$,
the set $\Pi_j \Lambda$ is the $j^{\rm th}$-component rectangle of $\Lambda$ in $\R^d$.
We define $\Pi \Lambda$ be the subset of $\R^d$ defined by
\beq\label{eq:pi-lambda1}
\Pi \Lambda \equiv \bigcup_{j=1}^N \Pi_j \Lambda \subset \R^d.
\eeq
For a  subset $\mathcal{J} \subset \{ 1, \ldots, N \}$, we define
\beq\label{eq:pi-lambda2}
\Pi_{\mathcal{J}} \Lambda = \bigcup_{j \in \mathcal{J}} \Pi_j \Lambda \subset \R^d.
\eeq
We also need the notion of an extension of $\Lambda$ that takes into account the size of the support of the single site potential $u$.
Suppose that $\supp u \subset B(0,R)$, for some $R > 0$. We define the extension $\hat{\Lambda}$
to be the $N$-particle rectangle containing $\Lambda$
obtained by replacing $\Lambda_{L_j}$ by $\Lambda_{L_j + 2R}$ in each component rectangle.

\begin{definition}\label{defn:r-separated}
Two rectangles $\Lambda, \Lambda' \subset \R^{Nd}$ are $R$-separated if there exists a nonempty subset $\mathcal{J}\subset \{ 1, \ldots, N \}$
of indices so that either
$$
{\rm dist} \left[  \Pi_{\mathcal{J}} \hat{\Lambda}, \Pi_{\mathcal{J}^c} \hat{\Lambda} \cup \Pi \hat{\Lambda'} \right] > 2R,
$$
or
$$
{\rm dist} \left[ \Pi_{\mathcal{J}} \hat{\Lambda'}, \Pi_{\mathcal{J}^c} \hat{\Lambda'} \cup \Pi \hat{\Lambda} \right] > 2R,
$$
where the distance is the Euclidean distance.
\end{definition}

The notion of $R$-separation guarantees that there are random variables in one rectangle that are independent of the random variables in the second rectangle. In the example above, the second condition of Definition \ref{defn:r-separated} holds with $R=1$ and $\mathcal{J} = \{ 2 \}$.

\begin{theorem}\label{thm:two-region1}
Let $H_{\omega,N}^\Lambda$ be an $N$-particle Hamiltonian restricted to an $N$-cube $\Lambda$. Suppose that $\supp u \subset B(0,R)$, for some $R > 0$. Let $I_0 = (- \infty, E_0] \subset \R$ be an energy interval
for any $E_0 > 0$ fixed.
If $\Lambda$ and $\Lambda'$ are two $R$-separated $N$-particle cubes, then for any $\epsilon > 0$ sufficiently small,
\beq\label{eq:two-region1}
\P \{ {\rm dist}~( \sigma (H_{\omega,N}^\Lambda) \cap I_0, \sigma(H_{\omega,N}^{\Lambda'}) \cap I_0) < \epsilon \} \leq
C_W C(E_0, d, N, u, U) E_0^d |\Lambda'| |\Lambda| s(2 \epsilon ),
\eeq
where the constant $C(E_0, d, N, u, U) $ is defined in \eqref{eq:main1}.
\end{theorem}

\begin{proof}
Since $\Lambda$ and $\Lambda'$ are partially separated cubes, we may assume without loss of generality, according to Definition \ref{defn:r-separated}, that there is a collection of sites associated with a one-particle cube $\Gamma \subset {\Lambda}$
so that the associated random variables are independent of those in $\Lambda'$. Any operator $H_{\omega,N}^{\Lambda}$ is uniformly lower semi-bounded and has only finitely-many eigenvalues in $I_0$.
As in \cite[page 564]{bcss1}, we then can compute
\bea\label{eq:cond-prob1}
\lefteqn{\P \{ {\rm dist}~( \sigma (H_{\omega,N}^\Lambda) \cap I_0, \sigma(H_{\omega,N}^{\Lambda'}) \cap I_0 < \epsilon \} }  \nonumber \\
 &=& \E_{\Gamma^c} [ \P_{\Gamma} \{ {\rm dist}~( \sigma (H_{\omega,N}^\Lambda) \cap I_0, \sigma(H_{\omega,N}^{\Lambda'}) \cap I_0)
 < \epsilon ~|~ \{ \omega_k\}_{k \in \Gamma^c} \} ] \nonumber \\
 & = & \E_{\Gamma^c} [ \P_{\Gamma} \{ {\rm min}_{j,k} | E_j(\Lambda) - E_k (\Lambda') | < \epsilon,~~ E_j(\Lambda), E_k (\Lambda') \in I_0 ~|~
\{ \omega_k\}_{k \in \Gamma^c} \} ] \nonumber \\
 & \leq & ( C_W E_0^d | \Lambda'| ) ~\E_{\Gamma^c} \left[ \sup_{E \in I_0} \P_{\Gamma} \{ {\rm dist} ~ ( \sigma (H_{\omega,N}^\Lambda) \cap I_0, E ) < \epsilon ~|~ \{ \omega_k\}_{k \in \Gamma^c} \} \right] \nonumber \\
  & \leq & C_W C(E_0, d, N, u, U) E_0^d |\Lambda'| |\Lambda| s(2 \epsilon ).
\eea
On the fourth line of \eqref{eq:cond-prob1}, we used Weyl's law with constant $C_W$ to bound the number of eigenvalues of $H_{\omega,N}^{\Lambda'}$ less than $E_0$. To calculate the probability with respect to the random variables in the one-particle cube $\Gamma$, we used the fact that we can obtain a lower bound in the calculation \eqref{eq:lb-Npot3} without the sum over the coordinate index $i$. We can then proceed with the proof as in section \ref{sec:wegnerNbody} for $\epsilon > 0$ small enough. This proves \eqref{eq:two-region1}.
\end{proof}

These two-region Wegner estimates are necessary for the multi-scale analysis proof of localization, see, for example, \cite[section 5]{klein-nguyen2}.
\section{Application to the Delone-Anderson model}\label{sec:delone1}

The techniques we have developed above
can be used to extend the results of Rojas-Molina and Veseli\'c \cite{RM-V1} and Klein \cite{klein1}
to the following non-ergodic model based on a one-particle random potential of Delone-Anderson type.
For any pair of finite positive constants $(m,M)$,, with $0 < m < M < \infty$, an $(m,M)$-Delone set
$\Gamma_{m,M}$ is a discrete subset $\{ z_j \} \subset \R^d$ having the property
that any cube of side length $m$ contains no more than one point, and any cube of side length $M$ contains at least one point. Hence the lattice
$M \Z^d$ contains at least one point of the Delone set $\Gamma_{m,M}$ in each cube of side length $M$. Consequently,
we can apply the above methods to the lattice $M \Z^d$. We decompose
the $(m,M)$-Delone set into two components $\Gamma_{m,M} = \Gamma_1 + \Gamma_2$, where $\Gamma_1$ contains exactly one point $y_j$ in $\Lambda_M(j)$.
Furthermore, these points are separated by a
distance at least $m$. We decompose the random potential into two pieces according to this decomposition of the set:
\beq\label{eq:delone1}
V_\omega (x) = \sum_{j \in \Z^d} u(x-z_j) = V_{\Gamma_1} (x)  + V_{\Gamma_2}(x).
\eeq
where
\beq\label{eq:delone2}
V_{\Gamma_1} (x) =  \sum_{j \in \Z^d} u(x-y_j) .
\eeq
The potential $V_{\Gamma_2}$ is bounded and independent of $V_{\Gamma_1}$ so we can add it to the potential $U$ without any loss of generality. We use the random variables associated with $\Gamma_1$ in the spectral averaging. In this way, and using his sfUCPSP, Klein obtained an improvement of
\cite[Theorem 4.2]{RM-V1}. For the $N$-body case, we obtain the analog of Theorem \ref{thm:main1} for the $N$-body Delone-Anderson random Schr\"odinger operator.


\end{document}